\documentclass[runningheads,11points]{llncs}
\usepackage{amssymb,amsmath}
\usepackage{graphicx, color,xcolor}
\usepackage{lmodern}

\usepackage{microtype,xspace}

\usepackage{dsfont}

\newtheorem{claimb}{Claim}

\newcommand{\W}{{\sf W}\xspace}

\newcommand{\NP}{{\sf NP}\xspace}
\newcommand{\np}{{\sf NP}\xspace}

\newcommand{\FPT}{{\sf FPT}\xspace}
\newcommand{\fpt}{{\sf FPT}\xspace}
\newcommand{\XP}{{\sf XP}\xspace}
\newcommand{\ETH}{{\sf ETH}\xspace}
\newcommand{\SETH}{{\sf SETH}\xspace}

\newcommand{\tw}{{\sf tw}}

\newcommand{\rw}{{\sf rw}\xspace}

\newcommand{\cut}[1]{\left( #1,\overline{#1} \right)}

\renewcommand{\deg}{{\sf deg}\xspace}

\newcommand{\Ocal}{{\mathcal O}}
\newcommand{\Acal}{{\mathcal A}}

\newcommand{\sbullet}{\,\begin{picture}(-1,1)(-1,-2.5)\circle*{2}\end{picture}\ }

\newcommand{\PartIntSup}[1]{\left\lceil #1\right\rceil}

\renewenvironment{proof}[1][]{\par \noindent {\bf Proof:#1}\ }{\hfill$\Box$\medskip}
\newenvironment{proofclaim}[1][]{\par \noindent {\bf Proof:#1}\ }{\hfill$\lrcorner$\medskip}

\newcommand{\mes}{{\sf mes}}
\newcommand{\mos}{{\sf mos}}

\newcommand{\Xo}{\chi_{{\sf odd}}}

\usepackage{cite}
\usepackage{hyperref}
\usepackage{boxedminipage}

\title{On the complexity of finding large odd induced subgraphs and odd colorings\thanks{Work supported by French projects DEMOGRAPH (ANR-16-CE40-0028), ESIGMA (ANR-17-CE23-0010), and ELIT (ANR-20-CE48-0008-01), the program ``Exploration Japon 2017'' of the French embassy in Japan, and the JSPS KAKENHI grant number JP18K11157. An extended abstract of this article appeared in the \emph{Proceedings of the 46th International Workshop on Graph-Theoretic Concepts in Computer Science (\textbf{WG}), volume 12301 of LNCS, pages 67-79, held online, June \textbf{2020}}. This article is permanently available at \url{https://arxiv.org/abs/2002.06078}.}
}

\authorrunning{R\'emy Belmonte and Ignasi Sau}
\titlerunning{Finding large odd induced subgraphs and odd colorings}

\author{R\'emy Belmonte\inst{1} and Ignasi Sau\inst{2}$^{,*}$}

\institute{
	University of Electro-Communications, Chofu, Japan\\
	\email{remy.belmonte@gmail.com}
\and
	LIRMM, Universit\'e de Montpellier, CNRS, Montpellier, France\\
		\email{ignasi.sau@lirmm.fr}
}

\begin{document}

\maketitle

\begin{abstract}
We study the complexity of the problems of finding, given a graph $G$, a largest induced subgraph of $G$ with all degrees odd (called an \emph{odd} subgraph), and the smallest number of odd subgraphs that partition $V(G)$. We call these parameters $\mos(G)$ and $\Xo(G)$, respectively. We prove that deciding whether $\Xo(G) \leq q$ is polynomial-time solvable if $q \leq 2$, and \NP-complete otherwise.
We provide algorithms in time $2^{\Ocal(\rw)} \cdot n^{\Ocal(1)}$ and $2^{\Ocal(q \cdot \rw)} \cdot n^{\Ocal(1)}$ to compute $\mos(G)$ and to decide whether $\Xo(G) \leq q$ on $n$-vertex graphs of rank-width at most $\rw$, respectively, and we prove that the dependency on rank-width is asymptotically optimal under the \ETH. Finally, we give some tight bounds for these parameters on restricted graph classes or in relation to other parameters.

\keywords{odd subgraph; odd coloring; rank-width; parameterized complexity; single-exponential algorithm;  Exponential Time Hypothesis.}
\end{abstract}

\section{Introduction}
\label{sec:intro}

Gallai proved, around 60 years ago, that the vertex set of every graph can be partitioned (in polynomial time) into two sets, each of them inducing a subgraph in which all vertices have even degree (cf.~\cite[Exercise 5.19]{Lov79}). Let us call such a subgraph an \emph{even} subgraph, and an \emph{odd} subgraph is defined similarly. Hence, every graph $G$ contains an even induced subgraph with at least $|V(G)|/2$ vertices. The analogous properties for odd subgraphs seem to be more elusive.  For a graph $G$, let $\mos(G)$ and $\Xo(G)$ be the order of a largest odd induced subgraph of $G$ and the minimum number of odd induced subgraphs of $G$ that partition $V(G)$, respectively. Note that for $\Xo(G)$ to be well-defined, each connected component of $G$ must have even order.

\smallskip

Concerning the former parameter, the following long-standing --and still open-- conjecture is cited as ``part of the graph theory folklore'' by Caro~\cite{Caro94}: there exists a positive constant $c$ such that every graph $G$ without isolated vertices satisfies $\mos(G) \geq c \cdot |V(G)|$. In the following discussion we only consider graphs without isolated vertices. Caro~\cite{Caro94} proved that $\mos(G) \geq (1 - o(1))\sqrt{n/6}$ where $n = |V(G)|$, and Scott~\cite{Scott92} improved this bound to $\frac{c n }{\log n}$ for some $c >0$. The conjecture has been proved for particular graph classes, such as trees~\cite{RadcliffeS95}, graphs of bounded chromatic number~\cite{Scott92}, graphs of maximum degree three~\cite{Berman0W97}, and graphs of tree-width at most two~\cite{HouYLL18}, also obtaining  best possible  constants.

As for the complexity of computing $\mos(G)$, Cai and Yang~\cite{CaiY11} studied, among other problems, two parameterized versions of this problem, and their reductions imply that it is \NP-hard. They also prove the \NP-hardness of computing the largest size of an even induced subgraph of a graph $G$, denoted $\mes(G)$. As a follow-up of~\cite{CaiY11}, related problems were studied by Cygan et al.~\cite{CyganMPPS14} and Goyal et al.~\cite{GoyalMPPS18}.

\smallskip

The parameter $\Xo$, which we call the \emph{odd chromatic number}, has attracted much less interest in the literature. To the best of our knowledge, it has only been considered by Scott~\cite{Scott01}, who defined it (using a different notation) and proved that the necessary condition discussed above for $\Xo(G)$ to be well-defined is also sufficient. He also provided lower and upper bounds on the maximum value of $\Xo(G)$ over all $n$-vertex graphs. In particular, there are graphs $G$ for which $\Xo(G) = \Omega(\sqrt{n})$.

\medskip
\noindent \textbf{Our contribution.} In this article we mostly focus on computational aspects of the parameters $\mos$ and $\Xo$. Note that, given a graph $G$, deciding whether $\Xo(G) \leq 1$ is trivial. We prove that deciding whether $\Xo(G) \leq q$ is \NP-complete for every $q \geq 3$ using a reduction from $q$-\textsc{Coloring}. We obtain a dichotomy on the complexity of computing $\Xo$ by showing that deciding whether $\Xo(G) \leq 2$ can be solved in polynomial time, through a reduction to the existence of a feasible solution to a system of linear equations over {\sf GF}[2].

Given the \NP-hardness of computing both parameters, we are interested in its parameterized complexity~\cite{DF13,FPT-book}, namely in identifying relevant parameters $k$ that allow for \FPT algorithms, that is, algorithms running in time $f(k) \cdot n^{\Ocal(1)}$ for some computable function $f$. Since the natural parameter, that is, the solution size, for $\mos$ has been studied by Cai and Yang~\cite{CaiY11} (and its dual as well), and  for $\Xo$ the problem is para-\NP-hard by our hardness results, we rather focus on {\sl structural parameters}. Two of the most successful ones are definitely tree-width and clique-width, or its parametrically equivalent parameter \emph{rank-width} introduced by Oum and Seymour~\cite{OumS06}. This latter parameter is {\sl stronger} than tree-width, in the sense that graph classes of bounded tree-width also have bounded rank-width; see Fig~\ref{fig:parameters}. We present algorithms running in time $2^{\Ocal(\rw)} \cdot n^{\Ocal(1)}$ for computing $\mes(G)$ and $\mos(G)$ for an $n$-vertex graph $G$ given along with a decomposition tree of width at most $\rw$, and an algorithm in time $2^{\Ocal(q \cdot \rw)} \cdot n^{\Ocal(1)}$ for deciding whether $\Xo(G) \leq q$. These algorithms are inspired by the ones of Bui-Xuan et al.~\cite{Bui-XuanTV10,Bui-XuanTV11} to solve \textsc{Maximum Independent Set} parameterized by rank-width and boolean-width, respectively. To the best of our knowledge, our algorithms are the first ones parameterized by rank-width for an \NP-hard problem running in time $2^{o(\rw^2)} \cdot n^{\Ocal(1)}$~\cite{abs-1805-11275,Bui-XuanTV10,Oum17,GanianHO13a,GanianHO13b}.

\begin{figure}[h!]
  \begin{center}
  \vspace{-.25cm}
    \includegraphics[width=.4\textwidth]{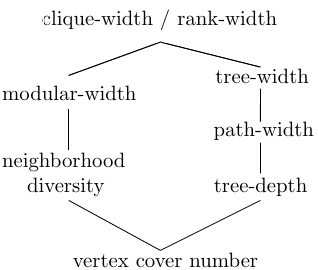}
  \end{center}\vspace{-.25cm}
  \caption{Diagram of some famous graph parameters (not defined here), where a parameter $\pi$  is above a parameter $\pi'$ if there exists a function $f: \mathds{N} \to \mathds{N}$ such that, for every integer $k \geq 0$, $\{ G \mid \pi'(G) \leq k \} \subseteq \{ G \mid \pi(G) \leq f(k) \}$.}
  \label{fig:parameters}
  \vspace{-.25cm}
\end{figure}

We also show that the dependency on rank-width of the above algorithms is asymptotically optimal under the Exponential Time Hypothesis (\ETH) of Impagliazzo et al.~\cite{ImpagliazzoPZ01,ImpagliazzoP01}.  For this, it suffices to obtain a {\sl linear} \NP-hardness reduction from a problem for which a subexponential algorithm does not exist under the \ETH. While our reduction to decide whether $\Xo(G)\leq q$ already satisfies this property, the \NP-hardness proof of Cai and Yang~\cite{CaiY11} for computing $\mes(G)$ and $\mos(G)$, which is from the \textsc{Exact Odd Set} problem~\cite{DowneyFVW99}, has a {\sl quadratic} blow-up, so only a lower bound of $2^{o(\sqrt{n})}$ can be deduced from it. Motivated by this, we present linear \NP-hardness reductions from  \textsc{2in3-Sat} to the problems of computing $\mes(G)$ and $\mos(G)$. The reduction itself is not very complicated, but the correctness proof requires some non-trivial arguments\footnote{We would like to mention that another \NP-hardness proof for computing $\mes(G)$ has very recently appeared online~\cite{Roy19}. The proof uses a chain of reductions from \textsc{Maximum Cut} and, although it also involves a quadratic blow-up, it can be avoided by starting from \textsc{Maximum Cut} restricted to graphs of bounded degree.}.

Finally, motivated by the complexity of computing these parameters, we obtain two tight bounds on their values. We first prove that for every graph $G$ with all components of even order, $\Xo(G) \leq \tw(G)+1$, where $\tw(G)$ denotes the tree-width of $G$. This result improves the best known lower bound on a parameter defined by Hou et al.~\cite{HouYLL18} (cf. Section~\ref{sec:bounds} for the details). On the other hand, we prove that, for every $n$-vertex graph $G$ such that $V(G)$ can be partitioned into two non-empty sets that are complete to each other (i.e., a {\sl join}), $\mos(G) \geq 2 \cdot  \PartIntSup{\frac{n-2}{4}}$.
In particular, this proves the conjecture about the linear size of an odd induced subgraph for {\sl cographs}, which are the graphs of clique-width two. This adds another graph class to the previous ones for which the conjecture is known to be true~\cite{RadcliffeS95,Scott92,Berman0W97,HouYLL18}.  It is interesting to mention that our proof implies that, for a cograph $G$, $\Xo(G) \leq3$, and this bound is also tight. While for cographs, or equivalently $P_4$-free graphs, we have proved that the odd chromatic number is bounded, we also show that it is unbounded for $P_5$-free graphs.

\medskip
\noindent \textbf{Organization.} We start with some preliminaries in Section~\ref{sec:prelim}. In Section~\ref{sec:hard} we provide the linear \NP-hardness reductions and the polynomial-time algorithm for deciding whether $\Xo(G) \leq 2$. The \FPT algorithms by rank-width are presented in Section~\ref{sec:DP}, and the tight bounds in Section~\ref{sec:bounds}. We conclude the article in Section~\ref{sec:concl} with a number of open problems and research directions. Additional results for related problems discussed in the conclusions can be found in Appendix~\ref{sec:DS}.

\section{Preliminaries}
\label{sec:prelim}
\noindent \textbf{Graphs}. We use standard graph-theoretic notation, and we refer the reader to~\cite{Diestel12} for any undefined notation. Let $G=(V,E)$ be a graph, $S \subseteq V$, and $H$ be a subgraph of $G$. We denote an edge between $u$ and $v$ by $uv$. The \emph{order} of $G$ is $|V|$. The \emph{degree} (resp. \emph{open neighborhood}, \emph{closed neighborhood}) of a vertex $v \in V$ is denoted by $\deg(v)$ (resp. $N(v)$, $N[v]$), and we let $\deg_H(v)= |N(v) \cap V(H)|$. We use the notation $G - S = G[V(G) \setminus S]$. The \emph{maximum} and \emph{minimum degree} of $G$ are denoted by $\Delta(G)$ and $\delta(G)$, respectively. We denote by $P_i$ the path on $i$ vertices. For two graphs $G_1$ and $G_2$, with $V(G_2) \subseteq V(G_1)$, the \emph{union} of $G_1$ and $G_2$ is the graph $(V(G_1), E(G_1) \cup E(G_2))$. The operation of \emph{contracting} an edge $uv$ consists in deleting both $u$ and $v$ and adding a new vertex $w$ with neighborhood $N(u) \cup N(v) \setminus \{u,v\}$. A graph $M$ is a \emph{minor} of $G$ if it can be obtained from a subgraph of $G$ by a sequence of edge contractions. For a positive integer $k  \geq 3$, the \emph{$k$-wheel} is the graph obtained from a
cycle $C$ on $k$ vertices  by adding a new vertex $v$ adjacent to all the vertices of $C$. A \emph{join} in a graph $G$ is a partition of $V(G)$ into two non-empty sets $V_1$ and $V_2$ such that every vertex in $V_1$ is adjacent to every vertex in $V_2$. For a positive integer $i$, we denote by $[i]$ the set containing every integer $j$ such that $1 \leq j \leq i$.

\medskip

\noindent \textbf{Parameterized complexity}. We refer the reader to~\cite{DF13,FG06,Nie06,FPT-book} for basic background on parameterized complexity, and we recall here only some basic definitions.
A \emph{parameterized problem} is a decision problem whose instances are pairs $(x,k) \in \Sigma^* \times \mathbb{N}$, where $k$ is called the \emph{parameter}.
A parameterized problem is \emph{fixed-parameter tractable} ({\sf FPT}) if there exists an algorithm $\Acal$, a computable function $f$, and a constant $c$ such that given an instance $I=(x,k)$,
$\Acal$ (called an {\sf FPT} \emph{algorithm}) correctly decides whether $I \in L$ in time bounded by $f(k) \cdot |I|^c$. A parameterized problem is \emph{slice-wise polynomial} ({\sf XP}) if there exists an algorithm $\Acal$ and two computable functions $f,g$ such that given an instance $I=(x,k)$,
$\Acal$ (called an {\sf XP} \emph{algorithm}) correctly decides whether $I \in L$ in time bounded by $f(k) \cdot |I|^{g(k)}$.


%

Within parameterized problems, the class {\sf W}[1] may be seen as the parameterized equivalent to the class \np of classical optimization problems. Without entering into details (see~\cite{DF13,FG06,Nie06,FPT-book} for the formal definitions), a parameterized problem being {\sf W}[1]-\emph{hard} can be seen as a strong evidence that this problem is {\sl not} \fpt. The canonical example of {\sf W}[1]-hard problem is \textsc{Independent Set} parameterized by the size of the solution.
To transfer ${\sf W}[1]$-hardness from one problem to another, one uses a \emph{parameterized reduction}, which given an input $I=(x,k)$ of the source problem, computes in time $f(k) \cdot |I|^c$, for some computable function $f$ and a constant $c$, an equivalent instance $I'=(x',k')$ of the target problem, such that $k'$ is bounded by a function depending only on $k$. An equivalent definition of $\W$[1]-hard problem is any problem that admits a parameterized reduction from \textsc{Independent Set} parameterized by the size of the solution.

The \emph{Exponential Time Hypothesis} (\ETH)  of Impagliazzo et al.~\cite{ImpagliazzoPZ01,ImpagliazzoP01} implies that the 3-\textsc{Sat} problem on $n$ variables cannot be solved in time $2^{o(n)}$. We say that a  polynomial reduction from a
 problem $\Pi_1$ to a problem $\Pi_2$, generating an input of size $n_2$ from an input of size $n_1$, is \emph{linear} if $n_2 = \Ocal(n_1)$. Clearly, if $\Pi_1$ cannot be solved, under the \ETH,  in time $2^{o(n)}$ on inputs of size $n$, and there exists a linear reduction from $\Pi_1$ to $\Pi_2$, then $\Pi_2$ cannot either.

\medskip

\noindent \textbf{Width parameters}. In this article we mention several width parameters of graphs, such as tree-width, rank-width, clique-width, or boolean-width. However, since we only deal with rank-width in our algorithms (cf. Section~\ref{sec:DP}), we give only the definition of this parameter here.

A \emph{decomposition tree} of a graph $G$ is a pair $(T,\delta)$ where $T$ is a full binary tree (i.e.,\ $T$ is rooted and every non-leaf node has two children) and $\delta$ is a bijection between the leaf set of $T$ and the vertex set of $G$.
For a node $w$ of $T$, we denote by $V_w$  the subset of $V(G)$ in bijection --via $\delta$-- with the leaves of the subtree of $T$ rooted at $w$. We say that the decomposition defines the \emph{cut} $\cut{V_w}$.
The \emph{rank-width} of a decomposition tree $(T,\delta)$ of a graph $G$, denoted by $\rw(T,\delta)$, is the maximum over all $w \in V(T)$ of the rank of the adjacency matrix of the bipartite graph $G[V_w,\overline{V_w}]$.
The \emph{rank-width of $G$}, denoted by $\rw(G)$, is  the minimum $\rw(T,\delta)$ over all decomposition trees $(T,\delta)$ of $G$.

\medskip

\noindent \textbf{Definition of the problems}. A graph is called \emph{odd} (resp. \emph{even}) if every vertex has odd (resp. even) degree. The \textsc{Maximum Odd Subgraph} (resp. \textsc{Maximum Even Subgraph}) problem consists in, given a graph $G$, determining the maximum order of an odd (resp. even) induced subgraph of $G$, that is, $\mos(G)$ (resp. $\mes(G)$). An \emph{odd $q$-coloring} of a graph $G=(V,E)$ is a set of $q$ odd induced subgraphs $H_1, \ldots, H_q$ of $G$ such that $V(H_1) , \cdots , V(H_q)$ is a partition of $V$. The \textsc{Odd $q$-Coloring} problem consists in determining whether an input graph $G$ admits an odd $q$-coloring. In the \textsc{Odd Chromatic Number} problem, the objective is to determine the smallest integer $q$ such that an input graph $G$ admits an odd $q$-coloring.
%
%
%
%

\section{Linear reductions and a polynomial-time algorithm}
\label{sec:hard}
We first present the linear reductions for \textsc{Maximum Even Subgraph} and \textsc{Maximum Odd Subgraph}, and then for \textsc{Odd $q$-Coloring} for $q \geq 3$.

\begin{theorem}\label{thm:hard-mos}
The \textsc{Maximum Even Subgraph} and \textsc{Maximum Odd Subgraph} problems are \NP-hard. Moreover, none of them can be solved in time  $2^{o(n)}$ on $n$-vertex graphs unless the \ETH fails.
\end{theorem}
\begin{proof}
The first statement has been already proved by Cai and Yang~\cite{CaiY11}, so we focus on the second one. We first deal with \textsc{Maximum Even Subgraph}, and we will then show how to deduce the hardness of \textsc{Maximum Odd Subgraph} with a simple modification.

In the \textsc{1in3-Sat} (resp. \textsc{2in3-Sat}) problem, we are given a \textsc{3-Sat} formula, and the objective is to decide whether there exists an assignment of the variables such that every clause contains {\sl exactly} one (resp. two) true literal(s). Porschen et al.~\cite[Lemma 5]{PorschenSSW14} showed that \textsc{1in3-Sat} is \NP-hard even if each clause contains exactly three variables and each variable occurs in exactly three clauses. Since their reduction from \textsc{3-Sat} is linear,  it follows that this restricted version of \textsc{1in3-Sat} cannot be solved in time $2^{o(n)}$ under the \ETH, where $n$ is the number of variables. By taking such an instance of \textsc{1in3-Sat} and building an equivalent instance of \textsc{2in3-Sat} by negating all the literals in every clause, it follows that \textsc{2in3-Sat} cannot be solved in time $2^{o(n)}$ under the \ETH, even if each clause contains exactly three variables and each variable occurs in exactly three clauses. We denote this version of \textsc{2in3-Sat} by $\textsc{2in3-Sat}_3$.

We proceed to present a linear reduction from $\textsc{2in3-Sat}_3$ to \textsc{Maximum Even Subgraph}. Given an instance $\varphi$ of $\textsc{2in3-Sat}_3$ with $n$ variables and $m$ clauses, we build an instance $G$ of \textsc{Maximum Even Subgraph} as follows (see Fig.~\ref{fig:reduction-mes} for an illustration). Let the variables and clauses of $\varphi$ be $x_1, \ldots, x_n$ and $c_1, \ldots, c_m$, respectively. Note that by the definition of the $\textsc{2in3-Sat}_3$ problem, we have that $n=m$.  Let $p \geq 88$ be a fixed even integer. For every variable $x_i$ of $\varphi$, we add to $G$ a variable gadget, with vertex set $X_i$, consisting of a path $P^i$ on $p$  vertices with endpoints $s_i$ and $t_i$, two vertices $x_i$ and $\bar{x}_i$ (corresponding to variable $x_i$ and its negation, respectively), and the five edges $s_ix_i$, $s_i\bar{x}_i$, $t_ix_i$, $t_i\bar{x}_i$, and $x_i\bar{x}_i$.

\begin{figure}[h!]
  \begin{center}
  \vspace{-.25cm}
    \includegraphics[width=.71\textwidth]{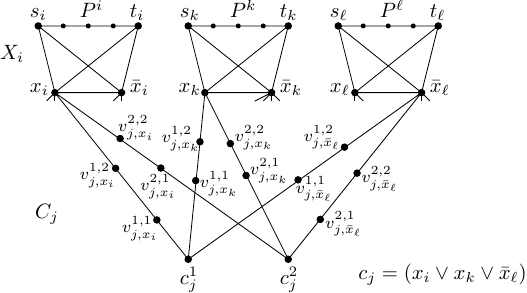}
  \end{center}\vspace{-.25cm}
  \caption{Graph $G$ built in the proof of Theorem~\ref{thm:hard-mos}, for a clause $c_j = (x_i \vee x_k \vee \bar{x}_{\ell})$.}
  \label{fig:reduction-mes}
  \vspace{-.25cm}
\end{figure}

For every clause $c_j = (\ell_1 \vee \ell_2 \vee \ell_3) $ of $\varphi$, where $\ell_1,\ell_2,\ell_3$ are the literals of $c_j$, we add to $G$ two vertices $c_j^1$ and $c_j^2$. For $i \in [3]$ and $k \in [2]$, we add a path on four vertices joining vertices $c_j^{k}$ and $\ell_i$, where the two internal vertices are new ones. For every such a path, we denote its internal vertices by $v_{j,\ell_i}^{k,1}$ and $v_{j,\ell_i}^{k,2}$, where $v_{j,\ell_i}^{k,1}$ is the one adjacent to $c_j^k$;  see Fig.~\ref{fig:reduction-mes}.
We denote by $C_j$ the set of 14 vertices of $G$ consisting of $c_j^1$, $c_j^2$, and the 12 internal vertices of the six paths joining them to the literals. This concludes the construction of $G$. Note that $|V(G)| = (p+16)n$, hence it is indeed a linear reduction. We claim that $\varphi$ is a positive instance of $\textsc{2in3-Sat}_3$ if and only if $\mes(G) \geq (p+13)n$.

Assume first that  that $\varphi$ is a positive instance of $\textsc{2in3-Sat}_3$, and let $\psi$ be the corresponding assignment of the variables. We proceed to define an even induced subgraph $H$ of $G$, with $|V(H)| = (p+13)n$, as follows. For every variable gadget $X_i$ of $G$, we include in $H$ the whole path $P^i$ and either $x_i$ or $\bar{x}_i$ depending on whether $\psi$ sets variable $x_i$ to true or false, respectively. For every clause $c_j = (\ell_1 \vee \ell_2 \vee \ell_3) $ of $\varphi$, suppose without loss of generality that $\psi$ sets $\ell_1$  and $\ell_2$ to true, and $\ell_3$ to false. We include in $H$ vertices $c_j^1, c_j^2$ and, out of the 12 internal vertices of the paths, we add to $H$ all of them except for the two vertices adjacent to $c_j^1$ and $c_j^2$ in the paths joining them to $\ell_3$, that is, the following set of 10 vertices: $\{v_{j,\ell_i}^{k,r}: i,k,r \in [2]\} \cup \{v_{j,\ell_3}^{1,2},v_{j,\ell_3}^{2,2}\}$.
See Fig.~\ref{fig:reduction-mes2} for an illustration with $\ell_3= x_k$. It can be verified that $H$ is indeed an even subgraph and that $|V(H)| = (p+1)n + 12n = (p+13)n$.

\begin{figure}[h!]
  \begin{center}
  \vspace{-.25cm}
    \includegraphics[width=.71\textwidth]{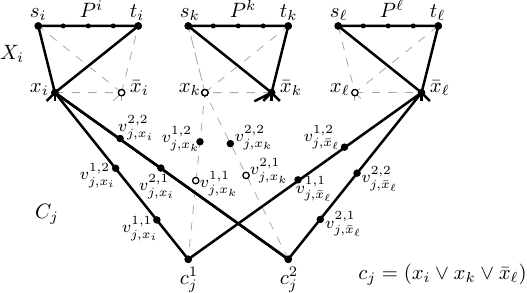}
  \end{center}\vspace{-.25cm}
  \caption{Example of an even subgraph $H$ in the graph of Fig.~\ref{fig:reduction-mes}, assuming that $\psi(x_i) = 1$ and $\psi(x_k)= \psi(x_{\ell}) =0$. Fat edges belong to $H$, while circled vertices do not.}
  \label{fig:reduction-mes2}
  \vspace{-.25cm}
\end{figure}

Conversely, suppose now that $G$ contains an even subgraph $H$ with $|V(H)| = \mes(G) \geq (p+13)n$. We state some properties of $H$ through a sequence of claims.

\begin{claimb}\label{claim:0}
For every $i \in [n]$, if $x_i \in V(H)$ and $\bar{x}_i \in V(H)$, then $|V(H) \cap X_i| \leq p/2 +2$, and if
$x_i \notin V(H)$ and $\bar{x}_i \notin V(H)$, then $|V(H) \cap X_i| \leq p/2$,
\end{claimb}
\begin{proofclaim}
Suppose first that $x_i \in V(H)$ and $\bar{x}_i \in V(H)$.
$H$ cannot contain any pair of adjacent vertices of $P^i$, as this would result in a vertex of degree one or three in $H$, which implies that $|V(H) \cap V(P^i)|\leq p/2$, and thus $|V(H) \cap X_i| \leq p/2 +2$. Suppose now that $x_i \notin V(H)$ and $\bar{x}_i \notin V(H)$. If $V(P^i) \subseteq V(H)$, then $\deg_H(s_i)=1$, so similarly as before we get that $|V(H) \cap X_i| \leq p/2$.
 \end{proofclaim}

\begin{claimb}\label{claim:1}
For every $i \in [n]$, either $x_i \in V(H)$ or $\bar{x}_i \in V(H)$.
\end{claimb}
\begin{proofclaim}
Assuming that the claim is not true, we will build from $H$  another even induced subgraph $H'$ of $G$ with $|V(H')| > |V(H)|$, contradicting the fact that $|V(H)| = \mes(G)$. For $i \in [n]$, let $J^G_i \subseteq [n]$ (resp. $J^H_i \subseteq [n]$) be the set of indices $j$ such that there exists at least one edge in $G$ (resp. in $H$) between $X_i$ and $C_j$. We define $H'$ according to the following iterative procedure, starting with $H'=H$:
\begin{enumerate}
\item For every  $i \in [n]$ such that exactly one of $x_i$ and $\bar{x}_i$ belongs to $H$, say $\ell_i$,  
    we ``double'' the paths from $X_i$ to $C_j$ for every  $j \in J^H_i$ (note that some of these paths may already be ``doubled'' in $H'$). Formally, for every $j \in J^H_i$, we replace $H'$ with $H' \cup G[\{\ell_i, v_{j,\ell_i}^{1,1},v_{j,\ell_i}^{1,2},v_{j,\ell_i}^{2,1},v_{j,\ell_i}^{2,2}, c_j^1,c_j^2\}]$.

\item For every  $i \in [n]$ such that $|V(H) \cap \{x_i , \bar{x}_i\}| \in \{0,2\}$, we first delete from $H'$, if any, all the vertices in $C_j$ for every index $j \in J^G_i$, that is, we replace $H'$ with $H' - \bigcup_{j \in J^G_i}C_j$ . Finally, we replace $H'$ with $(H' - \{\bar{x}_i\}) \cup G[\{x_i\} \cup V(P^i)]$, that is, we delete $\bar{x}_i$ and we add $x_i$ (if it did not already belong to $H'$) and the whole path $P^i$.

\item Note that after Step~2 above, the degree in $H'$ of all vertices in $X_i$ is even for every $i \in [n]$. Note also that, since in $H$ all the degrees were even and in Step~1 only paths that already existed in $H$ were doubled,  no vertex in a set $C_j$ can have degree one in $H'$. However, it is possible that a vertex in a set $C_j$, namely both $c_j^1$ and $c_j^2$,  has degree three in $H'$ (recall that every clause contains exactly three literals). Let $j \in [n]$ be such an index, and let $i_j \in [n]$ be an arbitrarily chosen index such that $H'$ contains edges from literal $\ell_{i_j}$ to $C_j$.  We guarantee that all the degrees in $H'$ are even by removing from $H'$ the internal vertices of the paths from $\ell_{i_j}$ to $c_j^1$ and $c_j^2$, that is, we replace $H'$ with $H' - \{v_{j,\ell_{i_j}}^{1,1},v_{j,\ell_{i_j}}^{1,2},v_{j,\ell_{i_j}}^{2,1},v_{j,\ell_{i_j}}^{2,2}\}$.
\end{enumerate}
Let $H'$ be the subgraph of $G$ obtained at the end of the above procedure.  By the discussion above, $H'$ is indeed an even induced subgraph. It remains to prove that $|V(H')| > |V(H)|$. We analyze each of the three steps separately.

Step~1 only adds new vertices to $H'$, so we can focus on Steps~2 and~3.

Let $i \in [n]$ be an index considered in Step~2, and let $H_1$ and $H_2$ be the current graphs before and after applying the procedure for index $i$, respectively. Since $|V(H) \cap \{x_i , \bar{x}_i\}| \in \{0,2\}$, Claim~\ref{claim:0} implies that $|V(H_1) \cap X_i| = |V(H) \cap X_i| \leq p/2 +2$. On the other hand, from the definition of Step~2 it follows that $|V(H_2) \cap X_i| = p+1$. Since every variable appears in exactly three clauses of $\varphi$, $|J_i^G| = 3$, so we have that $\sum_{j \in J_i^G}|C_j| = 42$, so in the first part of Step~2 at most 42 vertices are removed from $H_1$ in order to obtain $H_2$. Therefore,
\begin{eqnarray*}
|V(H_2)| - |V(H_1)| & = &  (|V(H_2) \cap X_i| - |V(H_1) \cap X_i|)  - \sum_{j \in J_i^G}|V(H_1) \cap C_j|\\
 & \geq & (p+1) - (p/2 +2) - 42 = p/2 - 43 > 0,
\end{eqnarray*}
where we have used that $p \geq 88$.

Let $i \in [n]$ be an index considered in Step~3,  and let again $H_1$ and $H_2$ be the current graphs before and after applying the procedure for index $i$, respectively. Let $j \in [n]$ be such that  $\deg_{H_1}(c_j^1) = \deg_{H_1}(c_j^2) = 3$ and let $c_j = (\ell_1 \vee \ell_2 \vee \ell_3)$; see Fig.~\ref{fig:Step3} for an illustration of the analysis. Since necessarily $\deg_{H}(c_j^1) = \deg_{H}(c_j^2) = 2$, but each of $\ell_1,\ell_2,\ell_3$ has a neighbor in $C_j$ in the graph $H$ (as otherwise $\deg_{H_1}(c_j^1) < 3$), it follows that $ \{\ell_1,\ell_2,\ell_3,c_j^1,c_j^2\} \subseteq V(H)$ and that, out of the six possible paths from $ \{\ell_1,\ell_2,\ell_3\}$ to $\{c_j^1,c_j^2\}$, exactly four of them are entirely in $H$, while two of them are entirely outside of $H$. That is, $|V(H) \cap C_j | = 10$. On the other hand, the definition of Step~3 implies that $|V(H_2) \cap C_j | = |V(H_1) \cap C_j | - 4 = 10$.
Hence, since Step~3 only deletes vertices in $C_j$, and we have proved that $|V(H) \cap C_j | = |V(H_2) \cap C_j | = 10$, we conclude that the joint application of Steps~1 and~3 does not decrease $|V(H)|$.

\begin{figure}[h!]
  \begin{center}
  \vspace{-.15cm}
    \includegraphics[width=.97\textwidth]{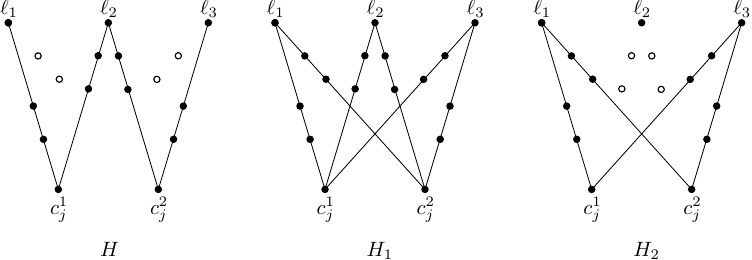}
  \end{center}\vspace{-.25cm}
  \caption{Subgraphs considered in the analysis of Step~3 in the construction of $H'$ in the proof of Claim~\ref{claim:1}. The circled vertices do {\sl not} belong to the corresponding graph.}
  \label{fig:Step3}
  \vspace{-.25cm}
\end{figure}


Note that, since we assume that the claim is not true, that is,  that there exists some index $i \in [n]$ such that $|V(H) \cap \{x_i , \bar{x}_i\}| \in \{0,2\}$, Step~2 in the construction of $H'$ has been applied at least once, therefore $|V(H')| > |V(H)|$ and the claim follows.
\end{proofclaim}

\begin{claimb}\label{claim:2}
For every $i \in [n]$, $V(P^i) \subseteq V(H)$.
\end{claimb}
\begin{proofclaim}
Consider an arbitrary  $i \in [n]$. By Claim~\ref{claim:1}, either $x_i \in V(H)$ or $\bar{x}_i \in V(H)$. Assume without loss of generality that $x_i \in V(H)$. If $V(P^i) \nsubseteq V(H)$, then the graph $H'$ defined as $H \cup G[\{x_i\} \cup V(P^i)]$ is an even induced subgraph of $G$ with $|V(H')| > |V(H)|$, a contradiction to the hypothesis that $|V(H)| = \mes(G)$.
\end{proofclaim}

\begin{claimb}\label{claim:3}
For every $j \in [n]$, $|V(H) \cap C_j| = 12$.
\end{claimb}
\begin{proofclaim}
If $H$ contained at least 13 vertices in $C_j$, then at least one of $c_j^1$ and $c_j^2$ would have degree three in $H$ (see Fig.~\ref{fig:reduction-mes}), a contradiction. On the other hand, Claims~\ref{claim:1} and~\ref{claim:2} imply that $|V(H) \cap \bigcup_{i \in [n]}X_i| = (p+1)n$. Since by hypothesis $|V(H)| \geq (p+13)n$, it follows that,  for every $j \in [n]$, $H$ contains exactly 12 vertices in $C_j$, and the claim follows. \end{proofclaim}


By Claim~\ref{claim:1}, the following assignment $\psi$ of the variables is well-defined: for $i \in [n]$, $\psi$ sets variable $x_i$ to true if and only if vertex $x_i$ belongs to $H$.
The following claim concludes the proof of the theorem for the even case.

\begin{claimb}\label{claim:4}
For every clause $c_j = (\ell_1 \vee \ell_2 \vee \ell_3) $ of $\varphi$, exactly two of its literals are set to true by $\psi$.
\end{claimb}
\begin{proofclaim}
By Claim~\ref{claim:3}, $|V(H) \cap C_j| = 12$, and then clearly $\deg_{H}(c_j^1)=\deg_{H}(c_j^2)=2$. Moreover, since $\deg_{H}(c_j^1)=\deg_{H}(c_j^2)=2$ and $\deg_H(v) \in \{0,2\}$ for every  $v \in C_j  \cap V(H)$, necessarily $C_j \setminus V(H) = \{ v_{j,\ell_i}^{1,1}, v_{j,\ell_i}^{2,1}  \}$ for some $i \in [3]$. Assume without loss of generality that $i=3$. Then it follows that $\ell_1, \ell_2 \in V(H)$ and that $\ell_3 \notin V(H)$. Indeed, if $\ell_1 \notin V(H)$ (the proof for $\ell_2$ is symmetric), the fact that $v_{j,\ell_1}^{1,1} \in V(H)$ implies that $\deg_H(v_{j,\ell_1}^{1,2})=1$, a contradiction. Similarly, if $\ell_3 \in V(H)$, the fact that $v_{j,\ell_3}^{1,1} \notin V(H)$ implies that $\deg_H(v_{j,\ell_3}^{1,2})=1$, a contradiction as well. Therefore, by the definition of $\psi$, it follows that exactly two of the literals of $c_j$ (namely, $\ell_1$ and $\ell_2$) are set to true by $\psi$. \end{proofclaim}


Note that the graph $G$ constructed above to prove the hardness of \textsc{Maximum Even Subgraph} has bounded maximum degree, namely $\Delta(G) \leq 9$.

\medskip
To prove the statement for \textsc{Maximum Odd Subgraph}, we present a simple linear reduction from  \textsc{Maximum Even Subgraph} that uses a trick of Cai and Yang~\cite[Theorem 4.5]{CaiY11}. Namely, let $G$ be an instance of \textsc{Maximum Even Subgraph} as constructed by the above reduction, and recall that $\mes(G) \leq (p+13)n =: k$. By adding an isolated vertex if needed, we may assume that $k$ is even. We build from $G$ an instance $G'$ of  \textsc{Maximum Odd Subgraph} by adding a ($k+1$)-wheel $W$ and making an arbitrary vertex of $W$ adjacent to all the vertices of $G$. Note that this is indeed a linear reduction. It can be easily checked that $\mos(G') \geq 2k+1$ if and only if $\mes(G) \geq k$, and the theorem follows.
\end{proof}

\begin{theorem}\label{thm:hard-qCol}
For every integer $q \geq 3$, given a graph $G$ on $n$ vertices, determining whether
$\Xo(G) \leq q$ is \NP-complete and, moreover, cannot be solved in time $2^{o(n)}$ unless the \ETH fails.
\end{theorem}
\begin{proof}
Membership in \NP is clear. For every integer $q \geq 3$, we present a linear reduction from the \textsc{$q$-Coloring} problem, which is well-known to be \NP-hard and not solvable in time $2^{o(n)}$ on $n$-vertex graphs unless the \ETH fails~\cite{ImpagliazzoPZ01,ImpagliazzoP01}. We will use the fact that any connected graph $G=(V,E)$ such that $|V| + |E|$ is even admits an orientation of $E$ such that, in the resulting digraph,  all the vertex in-degrees are odd; we call such an orientation an \emph{odd orientation}. Moreover, an odd orientation can be found in polynomial time (for a proof, see for instance~\cite{FrankJS99}).


Given an instance $G=(V,E)$ of \textsc{$q$-Coloring}, such that $G$ is connected, we build from $G$ an instance $G^{\sbullet}$ of \textsc{Odd $q$-Coloring} as follows. First, if $|V| + |E|$ is odd, we arbitrarily select a vertex $v \in V$ and add a triangle on three new vertices $v_1,v_2,v_3$ and the edge $vv_1$. Note that the resulting graph $G'=(V',E')$ is $q$-colorable for $q \geq 3$ if and only if $G$ is, and that $|V'| + |E'|$ is even. Hence, $E'$ admits an odd orientation $\phi$; see Fig.~\ref{fig:reduction3col}(a)-(b). We let $G^{\sbullet}$ be the graph obtained from $G'$ by subdividing every edge once; see Fig.~\ref{fig:reduction3col}(c). Note that the size of $G^{\sbullet}$ depends linearly on the size of $G$, as required.
We claim that $\chi(G) \leq q$ if and only if $\Xo(G^{\sbullet})\leq q$.

\begin{figure}[h!]
  \begin{center}
  \vspace{-.25cm}
    \includegraphics[width=.88\textwidth]{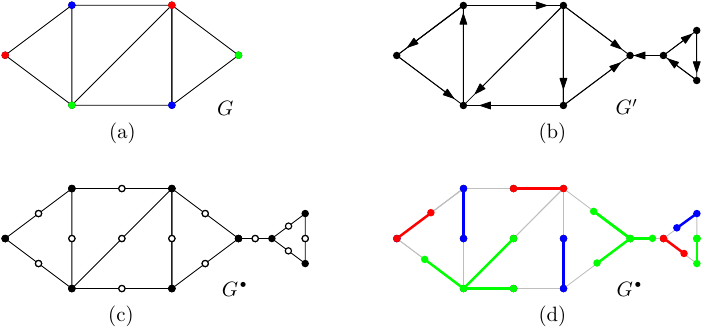}
  \end{center}\vspace{-.25cm}
  \caption{Example of the reduction of Theorem~\ref{thm:hard-qCol} with $q=3$: (a) Instance $G$ of \textsc{$3$-Coloring}, along with a proper 3-coloring. (b) Graph $G'$ obtained from $G$ along with an orientation of $E(G')$ with all in-degrees odd. (c) Instance $G^{\text{{\fontsize{4}{12}\selectfont$\bullet$}}}$ of \textsc{Odd $3$-Coloring}. (d) An odd $3$-coloring of $G^{\text{{\fontsize{4}{12}\selectfont$\bullet$}}}$.}
  \label{fig:reduction3col}
  \vspace{-.25cm}
\end{figure}

Assume first that we are given a proper $q$-coloring $c: V \to [q]$, which can trivially be extended to a proper $q$-coloring of $G'$. We define an odd $q$-coloring $c_{\sf odd}$ of $G^{\sbullet}$ as follows. If $v \in V(G^{\sbullet})$ is an original vertex of $V'$, we set $c_{\sf odd}(v) = c(v)$. Otherwise, if $v$ is a subdivision vertex between two vertices $u$ and $w$ of $V'$, we set $c_{\sf odd}(v) = c(u)$ if edge $uw$ is oriented toward $u$ in $\phi$, and $c_{\sf odd}(v) = c(w)$ otherwise; see Fig.~\ref{fig:reduction3col}(d). It can be easily verified that $c_{\sf odd}$ is indeed an odd $q$-coloring  of $G^{\sbullet}$.

Conversely, let $c_{\sf odd}: V(G^{\sbullet}) \to [q]$ be an odd $q$-coloring of $G^{\sbullet}$, let $uw$ be an edge of $G'$, and let $v$ be the subdivision vertex in $G^{\sbullet}$ between $u$ and $w$. If follows that $c_{\sf odd}(u) \neq c_{\sf odd}(w)$, as otherwise vertex $v$ would have degree zero or two in its color class. Therefore, letting $c(v) = c_{\sf odd}(v)$ for every vertex $v \in V(G)$ defines a proper $q$-coloring of $G$, and the theorem follows.\end{proof}


Theorem~\ref{thm:hard-qCol} establishes the \NP-hardness of \textsc{Odd $q$-Coloring} for every $q \geq 3$. On the other hand, the \textsc{Odd $1$-Coloring} is trivial, as for any graph $G$, $\Xo(G) \leq 1$ if and only if $G$ is an odd graph itself. Therefore, the only remaining case is \textsc{Odd $2$-Coloring}. In the next theorem we prove that this problem can be solved in polynomial time.

\begin{theorem}\label{thm:2Coloring}
The \textsc{Odd $2$-Coloring} problem  can be solved in polynomial time.
\end{theorem}
\begin{proof}
We will express the \textsc{Odd $2$-Coloring} problem as the existence of a feasible solution to a system of linear equations over the binary field, which can be determined in polynomial time using, for instance, Gaussian elimination. Given an instance $G=(V,E)$ of \textsc{Odd $2$-Coloring}, let its vertices be labeled $v_1, \ldots, v_n$. For every vertex $v_i \in V$ we create a binary variable $x_i$, and for every edge $v_iv_j \in E$, we create a binary variable $x_{i,j}$. The interpretation of these two types of variables is quite different. Namely, for a vertex variable $x_i$, its value corresponds to the color (either 0 or 1) assigned to vertex $v_i$. On the other hand, the value of an edge variable corresponds to whether this edge belongs to a monochromatic subgraph, that is, to whether both its endvertices get the same color. In this case, its value is 1, and 0 otherwise. We guarantee this latter property by adding the following set of linear equations:
\begin{equation}\label{eq:1}
x_i + x_j + x_{i,j} \equiv 1 \ \ \ \ \text{ for every edge $v_iv_j \in E$.}
\end{equation}

To guarantee that the degree of every vertex in each of the two monochromatic subgraphs is odd, we add  the following set of linear equations (for an edge variable $x_{i,j}$, to simplify the notation we interpret $x_{j,i} = x_{i,j}$):

\begin{equation}\label{eq:2}
\sum_{j: v_j \in N(v_i)}x_{i,j} \equiv 1 \ \ \ \ \text{ for every vertex $v_i \in V$.}
\end{equation}

Note that by Equation~(\ref{eq:1}), only monochromatic edges contribute to the sum of Equation~(\ref{eq:2}). Therefore, the above discussion implies that $\Xo(G) \leq 2$ if and only if the system of linear  equations given by Equations~(\ref{eq:1}) and~(\ref{eq:2}) admits a feasible solution, and the theorem follows.\end{proof}

\medskip

Note that the \textsc{Even $2$-Coloring} problem could be formulated in a similar way, just by replacing Equation~(\ref{eq:2}) with $\sum_{j: v_j \in N(v_i)}x_{i,j} \equiv 0$. However, this is not that interesting, since all the instances of \textsc{Even $2$-Coloring} are positive~\cite{Lov79}.

\section{Dynamic programming algorithms}
\label{sec:DP}
In this section, we present \FPT algorithms for \textsc{Maximum Odd/Even Subgraph} and \textsc{Odd} $q$-\textsc{Coloring}, parameterized by the rank-width of the input graph. The algorithms are similar to those of Bui-Xuan et al.~\cite{Bui-XuanTV10, Bui-XuanTV11} for \textsc{Maximum Independent Set} parameterized by rank-width and boolean-width, respectively, and even closer  to the one by Bui-Xuan et al.~\cite{Bui-XuanTV13} for so-called locally checkable vertex partitioning problems, in particular for {\sc Dominating Set}. There are however two key differences with our algorithms. First, while partial solutions for \textsc{Maximum Independent Set} are, themselves, independent sets, this is not true in general for odd subgraphs, where partial solutions may consist in a subgraph some vertices of which have even degree. Those vertices will impose some extra constraints on the remainder of the solution. The second difference is that, while the equivalence classes of~\cite{Bui-XuanTV10} and~\cite{Bui-XuanTV11} are based on neighborhoods of vertex sets, those for \textsc{Maximum Odd Subgraph} only require ``neighborhoods modulo 2''. This will allow us to consider only $2^{\Ocal(\rw)}$ equivalence classes, compared to $2^{\Ocal(\rw^2)}$ classes used in~\cite{Bui-XuanTV10} for \textsc{Maximum Independent Set}.

Throughout this section, we will rely on the notion of ``neighborhood modulo 2" of a set of vertices, defined as follows.
Given a graph $G$ and $X\subseteq V(G)$, the \emph{neighborhood of $X$ modulo 2}, denoted by $N_2(X)$, is the set $\triangle_{u\in X}(N(u))$, where the operator $\triangle$ denotes the symmetric difference.
Note that $N_2(X)$ is exactly the set of vertices in $V(G)\setminus X$ that have an odd number of neighbors in $X$. The results in this section are stated using the $\Ocal^*$ notation, which hides polynomial factors in the input size.

\begin{theorem}
\label{thm:odd-subgraph-rw}
Given a graph $G$ along with a decomposition tree of rank-width $\rw$, the \textsc{Maximum Odd Subgraph} problem can be solved in time $\Ocal^*(2^{3\rw})$.
\end{theorem}

\begin{proof}
We give a dynamic programming over the given decomposition tree $(T,L)$. Recall that there is a bijection between the leaves of $T$ and $V(G)$, and that each edge of $T$ corresponds to a cut $(A,\overline{A})$ of $G$. We begin by defining the equivalence relation over subsets of $A$, given a cut $(A,\overline{A})$: two sets $X, Y\subseteq V(G)$ are \textit{odd neighborhood equivalent} with regard to $A$, denoted by $X \equiv_2^A Y$, if $N_2(X)\setminus A = N_2(Y)\setminus A$. Then, given a row basis $\mathcal{B}$ of the adjacency matrix of $(A,\overline{A})$ over {\sf GF}[2],  where we interpret a vertex set as the vector corresponding to its vertices, we define the \emph{representative} of a set $X\subseteq A$ as the the unique set of vertices $R_A(X)\subseteq A$ such that $R_A(X)\subseteq \mathcal{B}$ and $X \equiv_2^A R_A(X)$. Observe that since $(A,\overline{A})$ is a cut of $(T,L)$, its adjacency matrix has rank at most $\rw(G)$, and therefore $|R_A(X)|\leq \rw(G)$. This implies, in particular, that there are at most $2^{\rw(G)}$ distinct representatives for subsets of a given set $A$.

We are now ready to define the tables of our algorithm. Given an edge $e$ of $(T,L)$ and its associated cut $(A,\overline{A})$ of $G$, we store in table $T_A$, for every pair $R,R'$ of representatives of subsets of $A$ and $\overline{A}$, respectively, a largest set $S\subseteq A$ such that $S$ is odd neighborhood equivalent to $R$, and all the vertices that have even degree in $G[S]$ is exactly the set $N_2(R') \cap S$. More formally:
\[(\maltese)\ T_A[R,R'] = \underset{S\subseteq A}{\text{maxset}}\{S \equiv_2^A R \wedge \{v\in S: |N(v)\cap S| \text{ is even}\} =  N_2(R') \cap S\}, \]
where the notation `maxset' indicates a largest set that satisfies the conditions.
In cases where edge $e$ is incident with a leaf, the cut associated with $e$ is of the form $(\{u\},V(G)\setminus \{u\})$. We assume $u$ is not an isolated vertex, as such vertices never belong to an odd subgraph.  We set $T_{\{u\}}[\emptyset,\emptyset] = T_{\{u\}}[\emptyset,\{v\}] = \emptyset$, and $T_{\{u\}}[\{u\},\{v\}] = \{u\}$, where $v$ is the unique vertex of a basis of the adjacency matrix of the cut $(V(G)\setminus \{u\}, \{u\})$, which is the only non-empty choice for $R'$.
The entry $T_{\{u\}}[\{u\},\emptyset]$ is left empty, since $\{u\}$ is the only set equivalent to $\{u\}$ and $u$ has even degree in the subgraph $G[\{u\}]$.

Given an edge $e$ of $(T,L)$ such that the tables of both edges incident with one endvertex of $e$, say $f,f'$, have been computed, we compute the table of $e$ as follows. Let us denote by $(A,\overline{A}), (X,\overline{X})$, and $(Y,\overline{Y})$ the cuts associated with $e,f$, and $f'$, respectively. For each pair of representatives $R_A, R_{\overline{A}}$ of the cut $(A,\overline{A})$, the value of $T_A[R_A, R_{\overline{A}}]$ is the largest $T_X[R_X, R_{\overline{X}}] \cup T_Y[R_Y, R_{\overline{Y}}]$, such that $R_X, R_{\overline{X}}, R_Y$, and $R_{\overline{Y}}$ satisfy the following conditions with regard to $R_A$ and $R_{\overline{A}}$:
\begin{itemize}
\item[(i)] $R_A \equiv_2^A R_X \triangle R_Y$,
\item[(ii)] $R_{\overline{X}} \equiv_2^{\overline{X}} R_{\overline{A}} \triangle R_Y$, and
\item[(ii')] $R_{\overline{Y}} \equiv_2^{\overline{Y}} R_{\overline{A}} \triangle R_X$.
\end{itemize}

We proceed with this computation, starting from the leaves, in a bottom-up manner, having previously rooted $T$ by choosing an arbitrary edge, subdividing it, and making the newly created vertex the root of $T$. Observe that in the final stage of the algorithm, when the tables of both edges $f,f'$ incident with the root have been computed, we compute the table for the root node as described above, with $\overline{A} = \emptyset$, since in this case $X \cup Y = V(G)$. Of the three conditions described above, condition~(i) becomes trivial, since $R_A=\emptyset$, and conditions (ii) and (ii') simplify to $R_{\overline{X}} \equiv_2^{\overline{X}} R_Y$, and $R_{\overline{Y}} \equiv_2^{\overline{Y}}R_X$, respectively. This immediately implies that, provided that the tables are updated correctly, the table for the root node will indeed contain an optimal solution.

We first observe that since, as noted above, there are at most $2^\rw$ representatives on each side of each cut, the choices of $R_X, R_Y$ and $R_{\overline{A}}$ uniquely determines $R_{\overline{X}},R_{\overline{Y}}$ and $R_A$ through equations (i), (ii), and (ii'), and computing new tables can be carried out in time $\Ocal^*(2^{3\rw})$. Since there are $O(|V(G)|)$ nodes in the decomposition, the running-time is $\Ocal^*(2^{3\rw})$ as well, as desired.

It now remains to prove that the algorithm correctly computes an optimal solution. The correctness of the tables for the leaves of $T$ follows from their description. We now prove by induction that the tables are correct for internal edges of $T$ as well. Let us assume $T_X$ and $T_Y$ have been fully and correctly computed for all possible representatives $R_X, R_{\overline{X}}, R_Y$, and $R_{\overline{Y}}$ as per the description above.

Finally, we argue that if $T_X$ and $T_Y$ are computed correctly, then so is $T_A$. Observe first that, for every $(R_X, R_{\overline{X}})$ and $(R_Y, R_{\overline{Y}})$ satisfying conditions (i),(ii) and (ii'), $S = T_X[R_X, R_{\overline{X}}] \cup T_Y[R_Y, R_{\overline{Y}}]$ satisfies $S \equiv_2^A R_A$ due to condition (I), and $\{v\in S: |N(v)\cap S| \text{ is even}\} =  N_2(R_{\overline{A}}) \cap S\}$ due to conditions (ii) and (ii'). In other words, $S$ satisfies $(\maltese)$, bar possibly the maximality condition.
To complete the proof that the tables are computed correctly, we argue that given any two representatives $R_A$ and $R_{\overline{A}}$ of $A$ and $\overline{A}$, respectively, there exist representatives $R_X, R_{\overline{X}}, R_Y$, and $R_{\overline{Y}}$ of $X, \overline{X}, Y$, and $\overline{Y}$, respectively,  that satisfy conditions (i), (ii), and (ii'), and such that $T_X[R_X, R_{\overline{X}}] \cup T_Y[R_Y, R_{\overline{Y}}]$ is a largest set that satisfies $(\maltese)$ with respect to $(R_A, R_{\overline{A}})$. Let $R_X, R_{\overline{X}}, R_Y$, and $R_{\overline{Y}}$ be representatives such that $T_A[R_A, R_{\overline{A}}] = T_X[R_X, R_{\overline{X}}] \cup T_Y[R_Y, R_{\overline{Y}}] = S$, and let $S_X$ and $S_Y$ denote $S \cap X$ and $S\cap Y$, respectively. Note that, since $X$ and $Y$ form a partition of $A$, $S_X$ and $S_Y$ form a partition of $S$, which implies $S = S_X \cup S_Y = S_X \triangle S_Y$.
We first show that $S$ indeed satisfies $(\maltese)$ with respect to $(A,\overline{A})$, i.e., $S \equiv_2^A R_A$ and $\{v\in S: |N(v)\cap S| \text{ is even}\} = N_2(R_{\overline{A}}) \cap S$.
For the first of those two conditions, combining it with the fact that $S=S_X\cup S_Y = S_X \triangle S_Y$, we only need to prove that $S_X \triangle S_Y \equiv_2^A R_X \triangle R_Y$.
Observe first that, since $X$ and $Y$ form a partition of $A$, we have that for every vertex $v\in \overline{A}, |N(v) \cap A| = |N(v) \cap X| +  |N(v) \cap Y|$. Therefore, for every sets $X',X''\subseteq X$ and $Y',Y''\subseteq Y$, it holds that if $X' \equiv_2^X X''$ and $Y' \equiv_2^Y Y''$, then $X' \triangle Y' \equiv_2^A X'' \triangle Y''$. From the definition of representative we obtain that $S \equiv_2^A S_X \triangle S_Y \equiv_2^A R_X \triangle R_Y$, as desired.

Let us now consider the second condition, i.e., $\{v\in S: |N(v)\cap S| \text{ is even}\} =  N_2(R_{\overline{A}}) \cap S$. Let $v$ be a vertex in $S_X$. If $|N(v)\cap S|$ is even, then at least one of the following cases holds:
\begin{itemize}
\item[$\bullet$] $|N(v)\cap S_X|$ is even and $v\not\in N_2(S_Y)$. Since $|N(v)\cap S_X|$ is even, we obtain from $(\maltese)$ in $T_X$ that $v\in N_2(R_{\overline{X}})$, which when combined with (ii) implies $v\in N_2(R_{\overline{A}}) \triangle N_2(S_Y)$. Since $v\not\in N_2(S_Y)$, it follows that $v\in N_2(R_{\overline{A}})$, as desired.
\item[$\bullet$] $|N(v)\cap S_X|$ is odd and $v\in N_2(S_Y)$. Symmetrically to the case above, we have that $v\not\in N_2(R_{\overline{X}})$, hence $v\not\in N_2(R_{\overline{A}}) \triangle N_2(S_Y)$ from  (ii), and since $v\in N_2(S_Y)$, it follows that $v\in N_2(R_{\overline{A}})$, as desired.
\end{itemize}
The case where $v\in S_Y$ is proved similarly, replacing condition (ii) with (ii'). Therefore, $\{v\in S: |N(v)\cap S| \text{ is even}\} \subseteq S \cap N_2(R_{\overline{A}})$.
Now, let  $v$ be a vertex in $S_X \cap N_2(R_{\overline{A}})$. From (ii), we obtain that $v\in N_2(S \cap \overline{X})$ if and only if $v\not\in N_2(S_Y)$. Since $T_X$ satisfies $(\maltese)$, it holds that $v\in N_2(S \cap \overline{X})$ if and only if $|N(v) \cap S_X|$ is even, and therefore $v\not\in N_2(S_Y)$ if and only if $|N(v)\cap S_X|$ is even. Therefore, $|N(v)\cap S| = |N(v)\cap S_X| + |N(v)\cap S_Y|$ is even, as desired. As above, the case where $v\in S_Y$ is proved similarly, replacing condition (ii) with (ii'). Therefore, $\{v\in S: |N(v)\cap S| \text{ is even}\} = S \cap N_2(R_{\overline{A}})$.

Finally, we prove the maximality of $S$ among all those sets that satisfy $(\maltese)$ with respect to $(R_A,R_{\overline{A}})$. Let us assume for a contradiction that there exists $S^*$ that satisfies $(\maltese)$ with respect to $(R_A,R_{\overline{A}})$ and such that $|S^*| > |S|$.
Let $S^*_X$ and $S^*_Y$ denote $S^* \cap X$ and $S^*\cap Y$, respectively. Observe that $S^*_X$ and $S^*_Y$ satisfy $(\maltese)$ with respect to some pairs of representatives $(R_X, R_{\overline{X}})$ and $(R_Y, R_{\overline{Y}})$, respectively. In addition, observe that since $S$ satisfies $(\maltese)$ with respect to $(R_A, R_{\overline{A}})$, it follows that $S, S_X$, and $S_Y$ satisfy conditions (i), (ii), and (ii') with respect to $(R_A,R_{\overline{A}})$, contradicting the assumption that $T_X$ and $T_Y$ were computed correctly.
\end{proof}

Small variations of the algorithm of Theorem~\ref{thm:odd-subgraph-rw} allow us to prove the following two theorems.

\begin{theorem}
\label{thm:even-subgraph-rw}
Given a graph $G$ along with a decomposition tree of rank-width $\rw$, the \textsc{Maximum Even Subgraph} problem can be solved in time $\Ocal^*(2^{3\rw})$.
\end{theorem}

\begin{proof}
The algorithm and its proof are nearly identical to the ones for \textsc{Maximum Odd Subgraph}, replacing condition $(\maltese)$ with
\[T_A[R,R'] = \underset{S\subseteq A}{\max}\{ |S| : S \equiv_2^A R \wedge \{v\in S: |N(v)\cap S| \text{ is odd}\} = S\cap N_2(R')\},\]
and the leaves are instead defined as  $T_{\{u\}}[\emptyset,\emptyset] = T_{\{u\}}[\emptyset,\{v\}] = \emptyset$ and $T_{\{u\}}[\{u\},\emptyset]=\{u\}$. $T_{\{u\}}[\{u\},\{v\}]$ is left empty, due to there being no odd subgraph of $G[\{u\}]$ with the same neighborhood as $\{u\}$ in $G - \{u\}$.
\end{proof}

\begin{theorem}
\label{thm:odd-chromatic-rw}
Given a graph $G$ along with a decomposition tree of rank-width $w$, the \textsc{Odd} $q$-\textsc{Coloring} problem can be solved in time $\Ocal^*(2^{\Ocal(q\cdot \rw)})$.
\end{theorem}

\begin{proof}
The algorithm is nearly identical to the one for \textsc{Maximum Odd Subgraph}, with the exception that the tables are indexed by $q$ pairs of representatives $R_1,R_1',\ldots,R_q,R_q'$, associating set $S_i$ with each pair $(R_i,R_i')$, with the additional constraint that $\bigcup_{i=1}^qS_i=A$. Since each table has at most $2^{2\cdot q\cdot \rw}$ entries, and computing a new table from two given ones takes time polynomial in the number of entries, we obtain the desired running time.
\end{proof} 

\section{Tight bounds}
\label{sec:bounds}
In this section we provide two tight bounds concerning odd induced subgraphs and odd colorings. Namely, we first provide in Theorem~\ref{thm:Xodd-tw} a tight upper bound on the odd chromatic number in terms of tree-width, and then we provide in Theorem~\ref{thm:tight-cographs} a tight lower bound on the size of a maximum odd induced subgraph for graphs that admit a join. 

\begin{theorem}\label{thm:Xodd-tw}
For every graph $G$ with all components of even order we have that $\Xo(G) \leq \tw(G)+1$, and this bound is tight.
\end{theorem}
\begin{proof}
Scott proved~\cite[Corollary 3]{Scott01} that every graph $G$ with all components of even order admits a vertex partition such that every vertex
class induces a tree with all degrees odd. Consider such a vertex partition, and let $G'$ be the graph obtained from $G$ by contracting each of the trees to a single vertex. Since $G'$ is a minor of $G$, we have that $\tw(G') \leq \tw (G)$. Now note that every proper vertex coloring of $G'$ using $q$ colors can be lifted to a partition of $V(G)$ into $q$ odd induced subgraphs (in fact, odd induced forests). Indeed, with every color $i$ of a proper $q$-coloring of $V(G')$ we associate an induced forest of $G$ defined by the union of the trees whose corresponding vertex in $G'$ is colored $i$; see Fig~\ref{fig:tw} for an example. Therefore,
$$
\Xo(G) \leq \chi(G') \leq \tw(G') + 1 \leq \tw (G) + 1,
$$
where we have used the well-known fact that the chromatic number of a graph is at most its tree-width plus one~\cite{Klo94}.

\begin{figure}[h!tb]
  \begin{center}
  \vspace{-.15cm}
    \includegraphics[width=.45\textwidth]{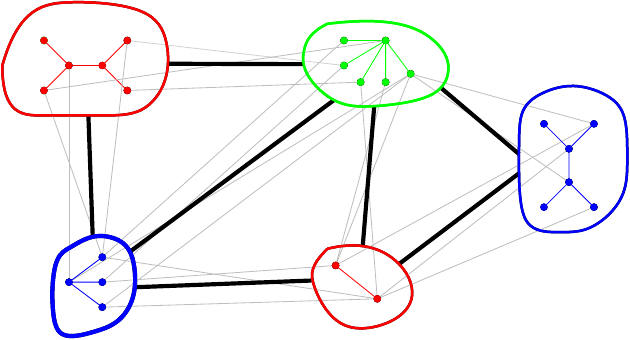}
  \end{center}\vspace{-.25cm}
  \caption{A graph $G$ and a partition of $V(G)$ into odd induced trees. Thick vertices and edges correspond to the contracted graph $G'$. A proper 3-coloring of $V(G')$ is depicted with colors.}
  \label{fig:tw}
  \vspace{-.25cm}
\end{figure}

To see that this bound it tight, consider a subdivided clique $K_n^{\sbullet}$, that is, the graph obtained from a clique on $n$ vertices, with $n \equiv0,3 \pmod{4}$, by subdividing every edge once; see Fig.~\ref{fig:example-sub} for an example. Since no pair of original vertices of the clique can get the same color, we have that $\Xo(K_n^{\sbullet}) = n = \tw(K_n^{\sbullet}) + 1$.
\end{proof}

\begin{figure}[h!tb]
  \begin{center}
  \vspace{-.15cm}
    \includegraphics[width=.27\textwidth]{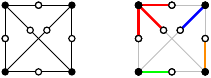}
  \end{center}\vspace{-.25cm}
  \caption{A subdivided $K_4$ and an odd $4$-coloring of it.}
  \label{fig:example-sub}
  \vspace{-.25cm}
\end{figure}


Let us mention some consequences of Theorem~\ref{thm:Xodd-tw}. Hou et al.~\cite{HouYLL18} define the following parameter. Let $\mathcal{G}_k$ be the set of all graphs of treewidth at most $k$ without isolated vertices, and let $c_k = \min_{G \in \mathcal{G}_k}\frac{\mos(G)}{|V(G)|}$. In~\cite{HouYLL18} the authors prove that $c_2 = 2/5$ and say that the best general lower bound is $c_k \geq \frac{1}{2k+2}$, which follows from a result of Scott~\cite{Scott92}. As an immediate corollary of Theorem~\ref{thm:Xodd-tw} it follows that $c_k \geq \frac{1}{k+1}$, which improves the lower bound by a factor two. As it is known~\cite{HouYLL18} that, for $k \in [4]$, $c_k \leq \frac{2}{k+3}$, our lower bound implies that $1/4 \leq c_3 \leq 1/3$ and $1/5 \leq c_4 \leq 2/7$.

\medskip

We now provide a lower bound on $\mos(G)$ for every graph that admits a join.

\begin{theorem}\label{thm:tight-cographs}
For every $n$-vertex graph $G$ that admits a join we have
$$
\mos(G) \geq 2 \cdot  \PartIntSup{\frac{n-2}{4}},
$$
and this bound is tight even for cographs.
\end{theorem}
\begin{proof}
Let $V_1,V_2 \subseteq V(G)$ define a join of $G$. We proceed to define a coloring $c:V(G) \to [3]$ such that, for some $i  \in [3]$, $G[c^{-1}(i)]$ is an odd subgraph with the claimed order. We distinguish three cases according to the parities of $n$, $|V_1|$, and $|V_2|$.
\begin{itemize}
\item[]\hspace{-.2cm}\textbf{Case 1}: $n$ is even and both $|V_1|,|V_2|$ are odd. Gallai proved (see~\cite{Caro94}) that the vertex set of every graph $G$ can be partitioned into two sets $A,B$ such that $G[A]$ is odd and $G[B]$ is even. We apply this result to both $G[V_1]$ and $G[V_2]$, yielding four sets $V_1^{{\sf e}},V_1^{{\sf o}},V_2^{{\sf e}},V_2^{{\sf o}}$ such that, for $i \in [2]$, $V_i = V_i^{{\sf e}} \uplus V_i^{{\sf o}}$, $G[V_i^{{\sf e}}]$ is even, and $G[V_i^{{\sf o}}]$ is odd. For $i \in [2]$, the fact that $G[V_i^{{\sf o}}]$ is odd implies that $|V_i^{{\sf o}}|$ is even, which in turn implies that $|V_i^{{\sf e}}|$ is odd since, by assumption, $|V_i^{{\sf e}}| + |V_i^{{\sf o}}| = |V_i|$ is odd. We define $c^{-1}(1) = V_1^{{\sf e}} \cup V_2^{{\sf e}}$, $c^{-1}(2) = V_1^{{\sf o}} \cup V_2^{{\sf o}}$, and $c^{-1}(3) = \emptyset$. Since $V_1,V_2 \neq \emptyset$ as they define a join of $G$, and both $|V_1|$ and $|V_2|$ are odd, it follows that, for $i \in [2]$, $V_i^{{\sf e}} \neq \emptyset$. This implies that both $G[c^{-1}(1)]$ and $G[c^{-1}(2)]$ are odd, and therefore one of them has order at least $n/2$.

\medskip
\item[]\hspace{-.2cm}\textbf{Case 2}: $n$ is even and both $|V_1|,|V_2|$ are even. Let $v_1 \in V_1$ and $v_2 \in V_2$ be two arbitrary vertices. We define $c(v_1) = c(v_2) =3$ and we apply Case~1 above to the graph $G - \{v_1,v_2\}$ with join given by $V_1 \setminus \{v_1\}$ and $V_2 \setminus \{v_2\}$, obtaining two odd induced subgraphs of $G$ colored 1 and 2, one of which has order at least $(n-2)/2$. Note that $G[c^{-1}(3)]$ is also an odd subgraph (an edge).

\medskip
\item[]\hspace{-.2cm}\textbf{Case 3}: $n$ is odd. Assume without loss of generality that $|V_1|$ is even and $|V_2|$ is odd, and let $v_1 \in V_1$ be an arbitrary vertex. We apply again Case~1 to the graph $G - \{v_1\}$ with join given by $V_1 \setminus \{v_1\}$ and $V_2$, obtaining two odd induced subgraphs of $G$, one of which has order at least $(n-1)/2$. Note that since $n$ is odd, in this case $\Xo(G)$ is not defined.
\end{itemize}

Summarizing, we have proved that if $n$ is even, then $\mos(G) \geq \frac{n-2}{2}$, and that if $n$ is odd, then $\mos(G) \geq \frac{n-1}{2}$. Taking into account that an odd subgraph must have even order, both cases imply that, for every $n$, $\mos(G) \geq 2 \cdot  \PartIntSup{\frac{n-2}{4}}$.

\medskip
Let us now see that this bound is tight for both even and odd values of $n$. For even $n$, consider the tripartite graph $K_{2,2,2}$. It can be checked that this graph contains none of $K_4$, $K_{1,3}$, and $2K_2$ (which are the only odd subgraphs on four vertices) as an induced subgraph, and therefore $\mos(K_{2,2,2}) = 2 = 2 \cdot  \PartIntSup{\frac{6-2}{4}}$. Finally, for odd $n$, consider $C_5^+$, that is, the graph obtained from $C_5$ by adding (any) two vertex-disjoint chords. Again, it can be checked that this graph contains none of $K_4$, $K_{1,3}$, and $2K_2$ as an induced subgraph, and therefore $\mos(C_5^+) = 2 = 2 \cdot  \PartIntSup{\frac{5-2}{4}}$.
Note that both $K_{2,2,2}$ and $C_5^+$ are cographs; see Fig.~\ref{fig:cographs}.
\end{proof}

\begin{figure}[h!tb]
  \begin{center}
  \vspace{-.15cm}
    \includegraphics[width=.47\textwidth]{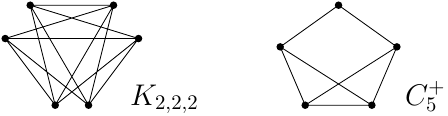}
  \end{center}\vspace{-.25cm}
  \caption{Cographs showing that the bound $2 \cdot  \PartIntSup{\frac{n-2}{4}}$ is tight.}
  \label{fig:cographs}
  \vspace{-.25cm}
\end{figure}

Determining a tight lower bound for cographs that are not necessarily connected remains open. The proof of Cases~1 and~2 of Theorem~\ref{thm:Xodd-tw} together with the fact that $\Xo(K_{2,2,2})=3$ (since $\mos(K_{2,2,2}) = 2$) yield the following corollary.

\begin{corollary}\label{cor:Xodd-cographs}
Let $G$ be a cograph with every connected component of even order. Then $\Xo(G) \leq 3$. Moreover, this bound is tight.
\end{corollary}

Note that cographs can be equivalently defined as $P_4$-free graphs. It is interesting to note that, in contrast to Corollary~\ref{cor:Xodd-cographs}, $P_5$-free graphs have unbounded odd chromatic number. Indeed, let $H_n$ be the graph obtained from the subdivided clique $K_n^{\sbullet}$, with $n \equiv0,3 \pmod{4}$, by adding an edge between each pair of original vertices of the clique. It can be checked that $\Xo(H_n) \geq n$ and, in fact, the proof of Theorem~\ref{thm:hard-qCol} implies that $\Xo(H_n) = n$. Note that $H_n$ is a split graph, hence split graphs have unbounded odd chromatic number. 

\section{Further research}
\label{sec:concl}
We considered computational aspects of the \textsc{Maximum Odd Subgraph} and \textsc{Odd $q$-Coloring} problems. A number of interesting questions remain open.

We gave in Theorem~\ref{thm:odd-chromatic-rw} an algorithm that solves  \textsc{Odd $q$-Coloring} in time $\Ocal^*(2^{\Ocal(q\cdot \rw)})$. Is the \textsc{Odd Chromatic Number} problem \FPT or $\W[1]$-hard parameterized by rank-width? A strongly related question is how the odd chromatic number depends on rank-width. We proved in Theorem~\ref{thm:Xodd-tw} that $\Xo(G) \leq \tw(G)+1$, but we do not know whether $\Xo(G) \leq f(\rw (G))$ for some function $f$. Note that this would not only yield an \FPT algorithm for \textsc{Odd Chromatic Number} by rank-width, but would also prove the conjecture about the linear size of a largest odd induced subgraph~\cite{Caro94} for all graphs of bounded rank-width. As a first step in this direction, we proved in Corollary~\ref{cor:Xodd-cographs} that cographs, which have rank-width at most one, have odd chromatic number at most three. It would be interesting to prove an upper bound for distance-hereditary graphs, which can be equivalently defined as graphs of rank-width one.

In fact, we do not even know whether \textsc{Odd Chromatic Number} by rank-width is in \XP. In view of
the algorithm of Theorem~\ref{thm:odd-chromatic-rw}, a sufficient condition for this would be that there exists a function $f$ such that
$\Xo(G) \leq f(\rw(G)) \cdot \log |V(G)|$ for every graph $G$ with all components of even order, but we were unable to prove it. Another plausible approach to obtain an \XP algorithm (which we believe to exist) would be to extend the general \XP algorithm of Rao~\cite{Rao07} parameterized by clique-width for vertex partitioning problems expressible in monadic second-order logic so to take into account the parities of the degrees, that is, to {\sl counting} monadic second order logic.

Toward an eventual $\W[1]$-hardness proof for \textsc{Odd Chromatic Number} by rank-width, a natural strategy is to try to adapt the reduction given by Fomin et al.~\cite{FominGLS10} to prove that \textsc{Chromatic Number} is $\W[1]$-hard by clique-width (hence, rank-width). This reduction is from \textsc{Equitable Coloring} parameterized by the number of colors plus tree-width, proved to be $\W[1]$-hard by Fellows et al.~\cite{FellowsFLRSST11}. By appropriately modifying the chain of reductions given in~\cite{FellowsFLRSST11}, we have only managed to prove that the naturally defined  \textsc{Odd Equitable Coloring} problem is $\W[1]$-hard by tree-width, but not if we add  the number of colors as a parameter.

Concerning \textsc{Odd $q$-Coloring} parameterized by tree-width, a straightforward dynamic programming algorithm that guesses, for every vertex, its color class and the parity of its degree within that class, runs in time $\Ocal^*((2q)^{\tw})$. Note that this algorithm together with Theorem~\ref{thm:Xodd-tw} yield an algorithm for \textsc{Odd Chromatic Number} in time $\Ocal^*((2\tw + 2)^{\tw})$. By the lower bound under the \ETH of Lokshtanov et al.~\cite{LokshtanovMS18ETH} for \textsc{Chromatic Number} by tree-width and the fact that our reduction of Theorem~\ref{thm:hard-qCol}  preserves tree-width, it follows that the dependency on tree-width of this algorithm is asymptotically optimal under the \ETH. It would be interesting to prove lower bounds under the \emph{Strong Exponential Time Hypothesis} (\SETH). Note that our reduction of Theorem~\ref{thm:hard-qCol} together with the lower bound under the \SETH of Lokshtanov et al.~\cite{LokshtanovMS18SETH} for \textsc{$q$-Coloring} by tree-width yield a lower bound for \textsc{Odd $q$-Coloring} of $\Ocal^*((q-\varepsilon)^{\tw})$ under the \SETH.

A natural direction is to study the complexity of  \textsc{Odd $q$-Coloring}, \textsc{Odd Chromatic Number}, \textsc{Maximum Even Subgraph}, and \textsc{Maximum Odd Subgraph} on restricted graph classes, such as split, interval, or chordal graphs. Concerning the \textsc{Even Subgraph} problem, given that in an $n$-vertex graph there always exists an even induced subgraph of size at least $n/2$~\cite{Lov79}, it makes sense to consider the parameterization of the problem above this lower bound, that is, ask for the existence of an even induced subgraph of size at least $n/2 + k$, $k$ being the parameter.

From a broader point of view, concerning problems that are \FPT by rank-width such as \textsc{Independent Set}, \textsc{Dominating Set}, \textsc{$q$-Coloring}, and  \textsc{Feedback Vertex Set}, the currently fastest algorithms run in time $\Ocal^*(2^{\Ocal(\rw^2)})$~\cite{Bui-XuanTV10,BergougnouxK19}, but the lower bounds under the  \ETH are just $\Ocal^*(2^{o(\rw)})$, from the classical linear \NP-hardness reductions from \textsc{3-Sat}. To improve the lower bounds, one should probably construct graph families of rank-width $\Ocal(\sqrt{n})$ and boolean-width $\Theta(n)$, as these problems can be solved in single-exponential time parameterized by boolean-width~\cite{Bui-XuanTV11,BergougnouxK19}. In particular, this would answer a question of Bui-Xuan et al.~\cite{Bui-XuanTV11} about the relation between rank-width and boolean-width. A related question is to determine which graphs, other than sparse ones, have rank-width $\Ocal(\sqrt{n})$.

Concerning problems that are $\W[1]$-hard by rank-width (or clique-width), such as  \textsc{Chromatic Number}, \textsc{Edge Dominating Set}, \textsc{Maximum Cut}, and \textsc{Hamiltonian Path}, the lower bounds that follow directly from the ones for clique-width~\cite{FominGLS10,FominGLSZ19,FominGLS14} leave a huge gap for rank-width. Closing this gap looks like a challenging problem.

Note that the problems that we considered can be seen as  the ``parity version'' of \textsc{Independent Set} and \textsc{$q$-Coloring}. It is natural to consider the parity version of other classical problems. In Appendix~\ref{sec:DS} we present some results on the parity version of domination problems, where the main contribution is to adapt the dynamic programming algorithms of Section~\ref{sec:DP} to these problems. It seems plausible that these results could also be lifted to the parity version of the vertex partitioning problems considered in~\cite{Bui-XuanTV13}.

\medskip

\noindent \textbf{Acknowledgement}. We would like to thank the anonymous referees for helpful and thorough comments that improved the presentation of the manuscript.

\bibliographystyle{abbrv}
\bibliography{biblio}

\newpage
\begin{appendix}
\section{Parity version of domination problems}
\label{sec:DS}
In this section we consider the ``parity version'' of domination problems. For simplicity, we just deal with the ``odd'' versions. Namely, an \emph{odd dominating set} (resp. \emph{odd total dominating set}) of a graph $G$ is a set $S \subseteq V(G)$  such that every vertex in $V(G) \setminus S$ (resp. $V(G)$) has an odd number of neighbors in $S$. Accordingly, in the \textsc{Minimum Odd Dominating Set} (resp. \textsc{Minimum Odd Total Dominating Set})  problem, we are given a graph $G$ and the objective is to find an odd dominating set (resp. odd total dominating set) in $G$ of minimum size.

It is worth mentioning that what is usually called an ``odd dominating set'' in the literature (cf. for instance~\cite{CaKl03} and the references given in~\cite{GravierJMP15,CaroKY05}) differs from the definition given above. Indeed, in previous work a set $S \subseteq V$ is said to be an odd dominating set if $|N[v] \cap S|$ is odd for {\sl every} vertex $v \in V(G)$. That is, vertices outside of $S$ must have an odd number of neighbors in $S$, while vertices in $S$ must have an even number of neighbors in $S$. Note that this definition differs from both definitions given in the above paragraph.

Concerning \textsc{Odd Total Dominating Set}, it was studied --among other parity problems-- by Halld{\'{o}}rsson et al.~\cite{HalldorssonKT00}, who proved its \NP-hardness by a reduction from the \textsc{Codeword of Minimal Weight} problem. However, the (quite involved) \NP-hardness proof of this latter problem by Vardy~\cite{Vardy97} involves several nonlinear blow-ups, so a lower bound of $2^{o(n)}$ under the \ETH cannot be deduced from it. Fortunately, we can indeed obtain a linear \NP-hardness reduction for both  \textsc{Minimum Odd Dominating Set} and \textsc{Minimum Odd Total Dominating Set} by doing simple local modifications  to the proof of Sutner~\cite[Theorem 3.2]{Sutner88a} for a variant of odd total domination, which is from the \textsc{$3$-Sat} problem. We omit the details here.

Once we know that none of these problems can be solved in time $2^{o(n)}$ on $n$-vertex graphs under the \ETH, our main contribution in this section is to adapt the dynamic programming algorithms presented in Section~\ref{sec:DP} to solve both \textsc{Minimum Odd Dominating Set} and \textsc{Minimum Odd Total Dominating Set} in single-exponential time parameterized by the rank-width of the input graph. Namely, we prove the following two results.

\begin{theorem}
\label{thm:odd-DS-rw}
Given a graph $G$ along with a decomposition tree of rank-width $w$, the \textsc{Minimum Odd Dominating Set} problem can be solved in time $\Ocal^*(2^{3\rw})$.
\end{theorem}

\begin{proof}
The algorithm and its proof are nearly identical to the ones for \textsc{Maximum Odd Subgraph} (cf. Theorem~\ref{thm:odd-subgraph-rw}), replacing condition $(\maltese)$ with
\[T_A[R,R'] = \underset{S\subseteq A}{\max}\{ |S| : S \equiv_2^A R \wedge \{v\in A\setminus S: |N(v)\cap S| \text{ is even}\} = A\cap N_2(R')\setminus S\},\] and the leaves are instead defined as  $T_{\{u\}}[\emptyset,\{v\}] = \emptyset$ and $T_{\{u\}}[\{u\},\emptyset]= T_{\{u\}}[\{u\},\{v\}] = \{u\}$. $T_{\{u\}}[\emptyset,\emptyset]$ is left empty.
\end{proof}

\begin{theorem}
\label{thm:odd-TDS-rw}
Given a graph $G$ along with a decomposition tree of rank-width $w$, the \textsc{Minimum Odd Total Dominating Set} problem can be solved in time $\Ocal^*(2^{3\rw})$.
\end{theorem}

\begin{proof}
Again, the algorithm and its proof are nearly identical to the ones for \textsc{Maximum Odd Subgraph} (cf. Theorem~\ref{thm:odd-subgraph-rw}), replacing condition $(\maltese)$ with
\[T_A[R,R'] = \underset{S\subseteq A}{\max}\{ |S| : S \equiv_2^A R \wedge \{v\in A: |N(v)\cap S| \text{ is even}\} = A\cap N_2(R')\},\] and the leaves are instead defined as  $T_{\{u\}}[\emptyset,\{v\}] = \emptyset$ and $T_{\{u\}}[\{u\},\{v\}] = \{u\}$. $T_{\{u\}}[\emptyset,\emptyset]$ and $T_{\{u\}}[\{u\},\emptyset]$ are left empty.
\end{proof}

As future work, it would be interesting to adapt the above algorithms to deal with the {\sl connected} version of both problems, where the (total) odd dominating set $S$ is further required to induce a connected graph; see~\cite{CaroKY05} for related work about this variant of domination.

%


%

\end{appendix}

\end{document}